\documentclass[10pt]{article}
\usepackage{enumerate}
\usepackage{amssymb}
\usepackage{xspace}
\usepackage{epsfig}
\usepackage{multirow}

\newtheorem{theorem}{Theorem}
\newtheorem{lemma}{Lemma}
\newtheorem{proposition}{Proposition}

\newenvironment{proof}{{\bf Proof:}}{\hfill\rule{2mm}{2mm}\\}

\newcommand{\ec}{\textsc{EC}\xspace}
\newcommand{\mec}{\textsc{EC($w$)}\xspace}
\newcommand{\bec}{\textsc{EC($b$)}\xspace}
\newcommand{\bmec}{\textsc{EC($w,b$)}\xspace}

\newcommand{\vc}{\textsc{VC}\xspace}
\newcommand{\mvc}{\textsc{VC($w$)}\xspace}
\newcommand{\bvc}{\textsc{VC($b$)}\xspace}
\newcommand{\bmvc}{\textsc{VC($w,b$)}\xspace}

\newcommand{\blvc}{\textsc{VC($\phi,b$)}\xspace}

\newcommand{\blec}{\textsc{EC($\phi,b$)}\xspace}

\newcommand{\algo}[1]{Algorithm \textsc{#1}}

\title{Bounded Max-Colorings of Graphs}

\author{
E. Bampis\thanks{IBISC, CNRS FRE 3190, Universit\'{e} d' \'{E}vry, France, \texttt{bampis@ibisc.univ-evry.fr}}
\and A. Kononov\thanks{Sobolev Institute of Mathematics, pr Koptyuga 4, Novosibirsk, Russia, \texttt{alvenko@math.nsc.ru}}
\and G. Lucarelli\thanks{Dept. of Informatics, Athens University of Economics and Business, Greece, \texttt{\{gluc,milis\}@aueb.gr}}
\and I. Milis$^\ddag$
}

\begin{document}

\maketitle

\begin{abstract}
In a bounded max-coloring of a vertex/edge weighted graph, each
color class is of cardinality at most $b$ and of weight equal to
the weight of the heaviest vertex/edge in this class. The bounded
max-vertex/edge-coloring  problems ask for such a coloring
minimizing the sum of all color classes' weights.

In this paper we present complexity results and approximation
algorithms for those problems on general graphs, bipartite graphs
and trees. We first show that both problems are polynomial for
trees, when the number of colors is fixed, and $H_b$ approximable
for general graphs, when the bound $b$ is fixed. For the bounded
max-vertex-coloring problem, we show a $17/11$-approximation
$\>$algorithm for bipartite graphs, a PTAS for trees as well as
for bipartite graphs when $b$ is fixed. For unit weights, we show
that the known $4/3$ lower bound for bipartite graphs is tight by
providing a simple $4/3$ approximation algorithm. For the bounded
max-edge-coloring problem, we prove approximation factors of
$3-2/\sqrt{2b}$, for general graphs, $\min\{e, 3-2/\sqrt{b}\}$,
for bipartite graphs, and $2$, for trees. Furthermore, we show
that this problem is NP-complete even for trees. This is the first
complexity result for max-coloring problems on trees.
\end{abstract}


\section{Introduction}

\label{section:introduction} The {\em bounded max-vertex-coloring}
(resp. {\em bounded max-edge-coloring}) problem takes as input a
graph $G=(V,E)$, a weight function $w: V \rightarrow N$ (resp. $w:
E \rightarrow N$) and an integer $b$; the question of this problem
is to find a  proper vertex- (resp. edge-) coloring of $G$,
$\mathcal{C}=\{C_1,C_2,\dots,C_k\}$, where each color $C_i$, $1
\leq i \leq k$, has weight $w_i = \max \{w(u)~|~ u \in C_i \}$
(resp. $w_i = \max \{w(e)~|~ e \in C_i \}$), cardinality $|C_i|
\leq b$, and the sum of colors' weights, $W=\sum_{i=1}^{k} w_i$,
is minimized.

We shall denote the vertex and edge bounded max-coloring problems
by \bmvc and \bmec, respectively. These problems, without the
presence of the cardinality bound $b$, have been already addressed
in the literature as {\em max-vertex-coloring} \cite{Pemmaraju04}
and {\em max-edge-coloring} \cite{Lucarelli09} problems; here we
denote them by \mvc and \mec, respectively. For unit weights we
get the {\em bounded vertex-coloring} \cite{Baker96} and {\em
bounded edge-coloring} \cite{Alon83} problems, denoted by \bvc and
\bec, respectively.  For both unbounded colors cardinalities and
unit weights, we have the classical {\em vertex-coloring} (\vc)
and {\em edge-coloring} (\ec) problems.\\

\noindent {\bf Motivation.} Max-coloring problems have been well
motivated in the literature. Max-vertex-coloring problems arise in
the management of dedicated memories, organized as buffer pools,
which is the case for wireless protocol stacks like \texttt{GPRS}
or \texttt{3G} \cite{Pemmaraju04,Pemmaraju05}.
Max-edge-coloring problems arise in switch based communication
systems, like \texttt{SS/TDMA} \cite{Bongiovanni81,Kesselman07},
where messages are to be transmitted through direct connections
established by an underlying network. Moreover, max-coloring
problems correspond to scheduling jobs with conflicts
into a batch scheduling environment \cite{Demange07,Finke08}.

In all applications mentioned above, context-related entities
require their service by physical resources for a time interval.
However, there exists in practice a natural constraint on the
number of entities assigned the same resource or different
resources at the same time. Indeed, the number of memory requests
assigned the same buffer is determined by strict deadlines on
their completion times, while the number of messages and jobs
assigned, at the same time, to different channels and machines,
respectively, is bounded by the number of the available resources.
The existence of such a constraint motivates the study of the
bounded max-coloring problems.\\


\noindent \textbf{Related Work.} It is well known that for general
graphs it is NP-hard to approxi- mate the \vc problem within any
constant factor and the \ec problem within a factor less than
$4/3$; for bipartite graphs both problems become polynomial.

The complexity of the \bvc problem (known also as Mutual Exclusion
Sche- duling problem \cite{Baker96}) on special graph classes has
been extensively studied (see \cite{Gardi08} and the references
therein). It is polynomial for trees \cite{Jarvis01}, but
NP-complete for bipartite graphs even for three colors
\cite{Bodlaender95}. This last result implies also a $4/3$
inapproximabilty  bound for the \bvc problem on bipartite graphs.

The \mvc problem is not approximable within a factor less than
$8/7$ even for planar bipartite graphs, unless P=NP
\cite{Demange07,Pemmaraju05}. This bound is tight for general
bipartite graphs, as an $8/7$-approximation algorithm is also
known \cite{deWerra08,Pemmaraju05}. On the other hand, the
complexity of the problem in trees is an open question, while a
PTAS for this case has been presented in
\cite{Pemmaraju05,Escoffier06}. Other results for the \mvc problem
on several graph classes have been also presented in
\cite{Demange07,deWerra08,Pemmaraju04,Pemmaraju05,Escoffier06,Epstein07}.

The \bec problem is polynomial for bipartite graphs
\cite{Bongiovanni81} as well as for general graphs if $b$ is fixed
\cite{Alon83}. Moreover, it is implied by the results in
\cite{Gardi08} that there is a $4/3$ approximation algorithm for
the \bec problem on general graphs.

The \mec problem is not approximable within a factor less than
$7/6$ even for cubic planar bipartite graphs with edge weights
$w(e) \in \{1,2,3\}$, unless P=NP \cite{deWerra08}. A simple
greedy $2$-approximation algorithm for general graphs has been
proposed in \cite{Kesselman07}.
The complexity of the \mec problem on trees remains also open,
while a $3/2$-approximation algorithm has been recently presented
\cite{Lucarelli09}.

Known results for the \bmvc and \bmec problems have been appeared
in the context of batch scheduling for complements of special
graph classes (see e.g. \cite{Finke08}). In this context both
problems have been shown to be polynomial for general graphs and
$b=2$ \cite{Boudhar00}.

In Table \ref{table:results} we summarize the best known results
for bipartite graphs and trees together with our contribution.

\begin{table}[htb]
\label{table:results}
\begin{center}
\begin{tabular}{|l|c|c|c|c|}
\hline
\multirow{2}{*}{Problem} & \multicolumn{2}{c|}{Bipartite graphs}            & \multicolumn{2}{c|}{Trees}                                   \\
\cline{2-5}
        & {~~ Lower Bound ~~}        & {Upper bound}                        & {Lower Bound}        & {Upper bound}                         \\
\hline \hline
\bvc    &  $4/3$ \cite{Bodlaender95} & $\bf 4/3$                            & \multicolumn{2}{c|}{$OPT$ \cite{Jarvis01}}                   \\
\hline
\mvc    & \multicolumn{2}{c|}{$8/7$ \cite{Demange07,deWerra08,Pemmaraju05}} & {open}$^*$           & $PTAS$ \cite{Pemmaraju05,Escoffier06} \\
\hline
\bmvc   & $4/3$ \cite{Bodlaender95}  & $\bf 17/11$                          & {open}$^*$           & $\bf PTAS$                            \\
\hline
\hline
\bec    & \multicolumn{2}{c|}{$OPT$ \cite{Bongiovanni81}}                   & \multicolumn{2}{c|}{$OPT$ \cite{Bongiovanni81}}              \\
\hline
\mec    & $7/6$ \cite{deWerra08}     & $2$ \cite{Kesselman07}               & {open}$^*$           & $3/2$ \cite{Lucarelli09}              \\
\hline
\bmec   & $7/6$ \cite{deWerra08}     & $\bf min\{e,3-2/\sqrt{b},H_b\}^{**}$ & \textbf{NP-complete} & $\bf 2$                               \\
\hline
\end{tabular}
\end{center}
\caption{Known and ours (in bold) approximability results for bounded and/or max coloring problems.
$^*$Even the complexity of the problem is unknown.
$^{**}$We also show a ratio of $\min\{3-2/\sqrt{2b},H_b\}$ for general graphs. In both cases,
the ratio $H_b$ holds only if $b$ is fixed.}
\end{table}

\noindent \textbf{Our results and organization of the paper.}
In this paper we deal with bounded max-coloring problems  on
general graphs, bipartite graphs and trees. Our interest in
bipartite graphs and trees is two-fold. Despite their simplicity,
these classes of graphs are important in their own right both from
a theoretical point of view but also from the applications'
perspective \cite{Kesselman07,Pemmaraju05}.

In the next section, we relate our problems with two well known
problems, namely the {\em list coloring} and the {\em set cover}
problems. We also introduce some useful notation. In  Section 3,
we deal with the \bmvc problem and we give a simple
$2$-approximation algorithm for bipartite graphs. As a byproduct,
we show that this  algorithm becomes a $4/3$-approximation
algorithm for the \bvc problem, which matches the $4/3$
$\>$inapproximability bound. Then, we present a generic scheme
that we show  to be a $17/11$-approximation algorithm for
bipartite graphs, while it becomes a PTAS for  trees as well as
for bipartite graphs when $b$ is a fixed constant. In Section 4,
we deal with the \bmec problem and we present approximation
algorithms of ratios $\min \{3-2/\sqrt{2b},H_b\}$, for general
graphs, and $\min\{e,3-2/\sqrt{b},H_b\}$, for bipartite graphs.
More interestingly, we prove that the \bmec problem on trees is
NP-complete. Given that the complexity question of \mvc, \mec and
\bmvc problems for trees remains open, this is the first
max-coloring problem on trees proven to be NP-hard. Finally, we
propose a $2$-approximation algorithm for the \bmec problem on
trees.


\section{Preliminaries and Notation}

We first establish a relation between our problems and  bounded
list coloring.\\

\noindent {\sc Bounded List Vertex} (resp. {\sc Edge}) {\sc Coloring problem}\\
\noindent \textsc{Instance:} A graph $G=(V,E)$, a set of colors
$\mathcal{C}=\{C_1,C_2,\dots,C_k\}$, a list of colors $\phi(u)
\subseteq \mathcal{C}$ for each $u \in V$ (resp. $\phi(e)
\subseteq \mathcal{C}$ for each $e \in E$), and integers
$b_i,~1 \leq i \leq k$.\\
\noindent \textsc{Question:} Is there a $k$-coloring of $G$ such
that each vertex $u$ (resp. edge $e$) is assigned a color in its
list $\phi(u)$ (resp. $\phi(e)$) and every color $C_i$ is used at
most $b_i$ times?\\

Clearly, the bounded list coloring problems, denoted by VC($\phi,
b_i$) and  EC($\phi, b_i$), are generalizations of the \bvc and
\bec problems, as well as of the \vc and \ec problems,
respectively. In the next theorem we summarize some of the known
results for the VC($\phi, b_i$) and EC($\phi, b_i$) problems that
we shall use in the rest of the paper.

\begin{theorem}
\label{th:list}
\begin{itemize}
\item []
\item[(i)] The VC($\phi, b_i$) problem is NP-complete even for chains, $|\phi(u)| \leq 2$, for all $u \in V$,
and $b_i \leq 5$, $1 \leq i \leq k$ \emph{\cite{Dror99}}.
\item [(ii)] Both VC($\phi, b_i$) and EC($\phi, b_i$) problems are polynomial for trees if the number of colors $k$
is fixed \emph{\cite{Gravier02,deWerra00}}.
\item [(iii)] The VC($\phi, b_i$) problem is polynomial for general graphs if $k=2$ \emph{\cite{Gravier02}}.
\end{itemize}
\end{theorem}

Using an exhaustive transformation of an instance of the \bmvc and
\bmec problems to an instance of the VC($\phi, b$) and EC($\phi,
b$) problems (where $b=b_i$, $1 \leq i \leq k$), respectively, and
Theorem \ref{th:list}(ii) (\cite{Gravier02,deWerra00}) we get next
proposition.

\begin{proposition} \label{prop:list}
For a fixed number of colors $k$, both the \bmvc and \bmec
problems on trees are polynomial.
\end{proposition}
\begin{proof}
We give here the proof for the \bmvc problem; the proof for the
\bmec problem is quite similar. Given a vertex weighted  graph
$G=(V,E)$ we generate all  $\left(|V| \atop k\right)$ possible
combinations for the weights of the $k$ colors. Let $w_1 \geq w_2
\geq \dots \geq w_k$ be the colors' weights in such a combination.
For each one of these combinations we construct an instance of the
\blvc problem on the graph $G$: is there a $k$-coloring of the
vertices of $G$ such that each color is used at most $b$ times and
each vertex $v \in V$ is assigned a color in $\phi(u)=\{C_i: w(u)
\leq w_i,~ 1 \leq i \leq k\}$? A ``yes'' answer to this instance
of the \blvc problem corresponds to a feasible solution for the
\bmvc problem of weight $W=\sum_{i=1}^k w_i$. An optimal solution
to the \bmvc problem corresponds to the combination where $W$ is
minimized.

There are $O(|V|^k)$ combinations of weights to be generated. For
a fixed $k$, by Theorem \ref{th:list}(ii)
(\cite{Gravier02,deWerra00}), the \blvc and \blec problems are
polynomial and the proposition follows.
\end{proof}

\noindent Next proposition is based on a relation between \bmvc
and \bmec problems and the set cover problem.

\begin{proposition}
\label{prop:setcover}
For a fixed bound $b$, there is an $H_b$-approximation
algorithm for both \bmvc and \bmec problems on general graphs.
\end{proposition}
\begin{proof}
In the set cover problem, we are given a universe $U$ of elements,
and a  collection, $\mathcal{S}=\{S_1,S_2,\dots,S_m\}$, of subsets
of $U$, each one of positive cost $c_i$, $1 \leq i \leq m$, and we
ask for a minimum cost subset of $\mathcal{S}$ that covers all
elements of $U$. For an instance of the \bmvc (resp. \bmec)
problem on a graph $G=(V,E)$ we consider the set of vertices $V$
(resp. edges $E$) corresponding to the universe set $U$. For each
$j$, $1 \leq j \leq b$, we generate all $\left(|V| \atop j\right)$
(resp. $\left(|E| \atop j\right)$) possible  subsets of
cardinality $j$ of vertices (resp. edges) of $G$. From all these
subsets we get rid those containing adjacent vertices (resp.
edges) and we consider the rest corresponding  to the set
$\mathcal{S}$. For each such subset $S_i \in S$ we set
$c_i=\max\{w(u)|u \in S_i\}$ (resp. $c_i=\max\{w(e)|e \in S_i\}$).

The cardinality of $S$ is $O(|V|^b)$, since $b$ is  $O(|V|)$, and
as an $H_b$-approximation algorithm is known for the set cover
problem \cite{Chvatal79}, the proposition follows.
\end{proof}

\noindent \textbf{Our notation.} Given a set $S$ and a positive
integer weight $w(s)$ for every element $s \in S$, we denote by
$\langle S \rangle = \langle s_1,s_2,\dots,s_{|S|} \rangle$ an
ordering of $S$ such that $w(s_1) \geq w(s_2) \geq \dots \geq
w(s_{|S|})$. For such an ordering of $S$ and a positive integer
$b$, let $k_S =\lceil\frac{|S|}{b}\rceil$. We define the
\emph{ordered $b$-partition} of $S$, denoted by $\mathcal{P}_S=
\{S_1, S_2, \dots, S_{k_{|S|}} \}$,  to be the partition of $S$
into $k_S$ subsets, such that $S_i
=\{s_j,s_{j+1},\dots,s_{\min\{j+b-1,|S|\}}\}$, $i=1,2,\dots,k$,
$j=(i-1)b+1$. In other words, $S_1$ contains the $b$ heaviest
elements of $S$, $S_2$ contains the next $b$ heaviest elements of
$S$ and so on; clearly, $S_{k_{|S|}}$ contains the $|S|~ mod ~b$
lightest elements of $S$.

By $OPT=w_1^*+w_2^*+\dots+w_{k^*}^*$ we denote the weight of an
optimal solution to the \bmvc or \bmec problem, where $w_i^*$, $1
\leq i \leq k^*$, is the weight of the $i$-th color class. By $\Delta$ we
denote the maximum degree of a graph.


\section{Bounded Max-Vertex-Coloring}
\label{section:bmvc}

In this section we first present a simple $2$-approximation
algorithm for the \bmvc problem on bipartite graphs. The
unweighted variant of this algorithm gives a $\frac{4}{3}$
approximation ratio for the \bvc problem on bipartite graphs,
which closes the approximability question for this case. Then, we
give a generic scheme which becomes a
$\frac{17}{11}$-approximation algorithm for bipartite graphs, a
PTAS for bipartite graphs and fixed $b$, as well as a PTAS for
trees. Recall also that by Proposition \ref{prop:setcover} there
is an $H_b$ approximation ratio for general graphs, if $b$ is
fixed.

\subsection{A simple split algorithm}

Let $G=(U \cup V,E)$, $|U \cup V|=n$, be a vertex weighted
bipartite graph. Our first algorithm colors the vertices of each
class of $G$ separately, by finding the ordered $b$-partitions of
classes $U$ and $V$. For the minimum number of colors $k^*$ it
holds that $k^* \geq \lceil\frac{|U|+|V|}{b}\rceil$ and,
therefore, $k=\lceil\frac{|U|}{b}\rceil+\lceil\frac{|V|}{b}\rceil
\leq \lceil\frac{|U|+|V|}{b}\rceil+1 \leq k^*+1$.\\

\noindent \algo{Split} \texttt{
\\1. Let $\mathcal{P}_U = \{U_1,U_2,\dots,U_{k_U}\}$ be the ordered $b$-partition of $U$;
\\2. Let $\mathcal{P}_V = \{V_1,V_2,\dots,V_{k_V}\}$ be the ordered $b$-partition of $V$;
\\3. Return the coloring $C= \mathcal{P}_U  \cup \mathcal{P}_V$;
}

\begin{theorem}
\label{theorem:bvmc} \algo{Split} returns a solution of weight $W
\leq 2 \cdot w_1^* +w_2^* + \dots + w_{k^*}^* \leq 2 \cdot OPT$
for the \bmvc problem in bipartite graphs.
\end{theorem}
\begin{proof}
Let $\langle C \rangle = \langle C_1,C_2,\dots,C_{k}\rangle$ be
the colors constructed by \algo{Split}, that is $w_1 \geq w_2 \geq
\dots \geq w_k$. Assume, w.l.o.g., that $U_x$, $1 \leq x \leq
k_U$, is the $i-th$ color in $\langle C \rangle$. Let also $u$ be
the heaviest vertex of $U_x$, that is $w(u)=w_i$.

The ordered $b$-partition of $U$ and $V$ implies that colors
$U_1,U_2,\dots,U_{x-1}$ and colors $V_1,V_2,\dots,V_y$, $y=i-x$,
appear before color $U_x$ in $\langle C \rangle$. Then, all $(x-1)
\cdot b$ vertices of colors $U_1,U_2,\dots,U_{x-1}$ are of weight
at least $w(v)$. Also, all $(y-1)\cdot b$ vertices of colors
$V_1,V_2,\dots,V_{y-1}$, are of weight at least the weight of the
heaviest vertex of color $V_y$ which is at least $w(v)$.
Therefore, there are in $G$ at least $(x-1) \cdot b + [(y-1) \cdot
b + 1] + 1=(x+y-2)\cdot b +2 = (i-2)\cdot b +2 $ vertices of
weight at least $w(v)=w_i$. In an optimal solution, these vertices
belong into at least $\lceil\frac{(i-2) \cdot b + 2}{b}\rceil =
(i-1)$ colors, each one of weight at least $w_i$. Hence,
$w_{i-1}^* \geq w_i,~ 2 \leq i \leq k$. Clearly, $w_1 = w_1^*$,
since both are equal to the weight of the heaviest vertex of the
graph, and as $k \leq k^*+1$, we obtain

\begin{center}
$W = \sum_{i=1}^k w_i
 = w_1^* + \sum_{i=2}^k w_i
 \leq w_1^* + \sum_{i=1}^{k-1} w_i^*
 \leq  w_1^* + \sum_{i=1}^{k^*} w_i^*$\\
 $= 2 \cdot w_1^* + w_2^* + \dots + w_{k^*}^*
 \leq 2 \cdot OPT$.
\end{center}
\end{proof}

The complexity of \algo{Split} is dominated by the sorting needed
to obtain the ordered $b$-partitions  of $U$ and $V$ in Lines 1
and 2, that is $O(n \cdot \log n)$.\\


\noindent \algo{Split} applies also to the \bvc problem on
bipartite graphs. Moreover, the absence of weights in the \bvc
problem allows a tight analysis with respect to the $\frac{4}{3}$
inapproximabilty bound.

\begin{theorem}
\label{theorem:bvc} There is a $\frac{4}{3}$-approximation
algorithm for the \bvc problem on bipartite graphs.
\end{theorem}
\begin{proof}
Assume, first, that $|U|+|V|\geq 2b+1$. Then, $k^* \geq
\lceil\frac{2b+1}{b}\rceil =3$ and, since $k \leq k^*+1$, we get
$\frac{k}{k^*} \leq \frac{4}{3}$.

Assume, next, that $b < |U|+|V| \leq 2b$. In this case the optimal
solution consists of two or three colors and it is polynomial to
decide between them. In fact, it is polynomial to decide if such a
bipartite graph can be colored with two colors even for the
generalized VC($\phi,b$) problem (see also Theorem
\ref{th:list}(iii) (\cite{Gravier02})).

Assume, finally, that  $|U|+|V| \leq b$. Then, an optimal solution
consists of either two colors (if $E \neq  \emptyset$) or one
color (if $E=\emptyset$).
\end{proof}

\subsection{A generic scheme}

To obtain our scheme we split a bipartite graph $G=(U \cup V,E)$,
$|U \cup V|=n$, into two subgraphs $G_{1,j}$ and $G_{j+1,n}$
induced by the $j$ heaviest and the $n-j$ lightest vertices of
$G$, respectively (by convention, we consider $G_{1,0}$ as an
empty subgraph). Our scheme depends on a parameter $p$ such that
all the vertices of $G$ of weights $w_1^*,w_2^*, \dots, w_{p-1}^*$
are in a subgraph  $G_{1,j}$. This is always possible for some
$j\leq b(p-1)$, since each color of an optimal solution for $G$
contains at most $b$ vertices. In fact, for every $j$, $1 \leq j
\leq b(p-1)$, we obtain a solution for the whole graph by
concatenating an optimal solution of at most $p-1$ colors for
$G_{1,j}$, if there is one, and the solution obtained by
\algo{Split} for $G_{j+1,n}$.\\

\noindent \algo{Scheme}$(p)$ \texttt{
\\1. Let $\langle U \cup V \rangle = \langle u_1,u_2,\dots u_n \rangle$;
\\2. For $j=0,1,\dots, b\cdot(p-1)$ do
\\3. $~$ Split the graph into two vertex induced subgraphs:
\\   $~~~~~$ - $G_{1,j}$ induced by vertices $u_1,u_2,\dots,u_j$
\\   $~~~~~$ - $G_{j+1,n}$ induced by vertices $u_{j+1},u_{j+2},\dots,u_n $
\\4. $~$ If there is a solution for $G_{1,j}$ with at most $p-1$ colors then
\\5. $~~~$ Find an optimal solution  for $G_{1,j}$ with at most $p-1$ colors;
\\6. $~~~$ Run \algo{Split} for $G_{j+1,n}$;
\\7. $~~~$ Concatenate the two solutions found in Lines 5 and 6;
\\8. Return the best solution found;
}

\begin{lemma}
\label{scheme} \algo{Scheme}(p) achieves a $(1+\frac{1}{H_p})$
approximation ratio for the \bmvc problem.
\end{lemma}
\begin{proof}
Consider the iteration $j$, $j \leq b\cdot(p-1)$, of the algorithm
where the weight of the heaviest vertex in $G_{j+1,n}$ equals to
the weight of the $i$-th color of an optimal solution, i.e.
$w(u_{j+1})=w_i^*$, $1 \leq i \leq p$.

The vertices of $G_{1,j}$ are a subset of those appeared  in the
$i-1$ heaviest colors of the optimal solution. Thus, an optimal
solution for $G_{1,j}$ is of weight $OPT_{1,j} \leq
w_1^*+w_2^*+\dots+w_{i-1}^*$.

The vertices of $G_{j+1,n}$ are a superset of those appeared in
the $k^*-(i-1)$ lightest colors of the optimal solution. The extra
vertices of  $G_{j+1,n}$ are of weight at most $w_i^*$ and appear
in an optimal solution into at most $i-1$ colors. Thus, an optimal
solution for $G_{j+1,n}$ is of  weight $OPT_{j+1,n} \leq
w_i^*+w_{i+1}^*+\dots+w_{k^*}^* + (i-1) \cdot w_i^* = i \cdot
w_i^*+w_{i+1}^*+\dots+w_{k^*}^*$. By Theorem \ref{theorem:bvmc},
\algo{Split} returns a solution for $G_{j+1,n}$ of weight
$W_{j+1,n} \leq (i+1) \cdot w_i^*+w_{i+1}^*+\dots+w_{k^*}^*$.

Therefore, the solution found in this iteration $j$ for the whole
graph $G$ is of weight $W_i =OPT_{1,j} + W_{j+1,n} \leq w_1^* +
w_2^* + \dots + w_{i-1}^* + (i+1) \cdot w_i^* + w_{i+1}^* + \dots
+ w_{k^*}^*$.

In all the iterations of the algorithm we obtain $p$ such
inequalities for $W$. By multiplying the $i$-th, $1 \leq i \leq
p$, inequality by $\frac{1}{i \cdot (H_p+1)}$ and adding up all of
them, we have $(\sum_{i=1}^p \frac{1}{i \cdot (H_p+1)}) \cdot W
\leq OPT$, that is $\frac{W}{OPT} \leq \frac{H_p+1}{H_p} =
1+\frac{1}{H_p}$.
\end{proof}

The complexity of the \algo{Scheme}(p) is $O(bp(f(p)+n \log n))$,
where $O(f(p))$ is the complexity of checking for the existence of
solutions with at most $p-1$ colors for $G_{1,j}$ and finding an
optimal one among them, while $O(n \log n)$ is the complexity  of
\algo{Split}. \algo{Scheme}(1) coincides with \algo{Split}.
\algo{Scheme}(2) has simply to check if the $j \leq b$ vertices of
$G_{1,j}$ are independent from each other and, therefore, it
derives a $\frac{5}{3}$ approximate solution in polynomial time.
\algo{Scheme}(3) has to check and find, a two color solution for
$G_{1,j}$, if any. This can be done in polynomial time  by Theorem
\ref{th:list}(iii) (\cite{Gravier02}). Thus, \algo{Scheme}(3) is a
polynomial time $\frac{17}{11}$-approximation algorithm for the
\bmvc problem on bipartite graphs.

However, when $p \geq 4$ and $b$ is a part of the instance,
finding an optimal  solution in $G_{1,j}$ is an NP-hard problem
(even for the \bvc problem \cite{Bodlaender95}). Hence, we
consider that $b$ is a fixed constant. In this case, we run an
exhaustive algorithm for finding, if any, an optimal solution in
$G_{1,j}$ of at most $p-1$ colors. The complexity of such an
exhaustive algorithm is  $O((p-1)^{b\cdot(p-1)})$ and thus, the
complexity of \algo{Scheme}$(p)$, $p \geq 4$, becomes
$O(bp^{bp}+n^2\log n)$, since $bp$ is $O(n)$. Choosing
$\epsilon=\frac{1}{H_p}$, we get $p=O(2^{\frac{1}{\epsilon}})$.
Consequently, for fixed $b$, we have a PTAS for the \bmvc problem
on bipartite graphs, that is an approximation ratio of
$1+\frac{1}{H_p}=1+\epsilon$ within
$O(b(2^{\frac{1}{\epsilon}})^{b2^{\frac{1}{\epsilon}}}+n^2\log n)$
time.

Furthermore, in the particular case of trees, checking the
existence of solutions with at most $p-1$ colors for $G_{1,j}$,
and finding an optimal one among them, can be done, by Proposition
\ref{prop:list}, in polynomial time for fixed $p$. The complexity
of our scheme in this case becomes
$O(b2^{\frac{1}{\epsilon}}(n^{2^{\frac{1}{\epsilon}}}+n^2 \log
n))$. Therefore, the following theorem holds.

\begin{theorem}
For the \bmvc problem, \algo{Scheme}(p) is a\\
(i) polynomial time $\frac{17}{11}$-approximation algorithm for bipartite graphs (for $p=3$),\\
(ii) PTAS for bipartite graphs if $b$ is fixed, \\
(iii) PTAS for trees.
\end{theorem}


\section{Bounded Max-Edge-Coloring}
\label{section:bmec}

In this section we deal with the complexity and approximability of
the \bmec problem. We present, first, approximation results for
general and bipartite graphs. Then, we prove that the problem is
NP-complete for trees and we give a 2-approximation algorithm for
this case.

\subsection{General and bipartite graphs}

We first adapt the greedy $2$-approximation algorithm presented in
\cite{Kesselman07} for the \mec problem to the \bmec problem.\\

\noindent \algo{Greedy} \texttt{
\\1. Let $\langle E \rangle = \langle e_1, e_2, \dots, e_{|E|}\rangle$;
\\2. For $j=1,2,\dots,|E|$ do
\\3. $~$ Insert edge $e_j$ in the first color of cardinality less than $b$
\\$~~~~~$ which does not contain other edges adjacent to $e_j$;\\
}

The analysis of \algo{Greedy} is based on tight bounds on the
number of colors in a solution to the \bmec problem.

\begin{proposition}
\label{prop:bmec_greedy} \algo{Greedy} achieves approximation
ratios  of $(3-\frac{2}{\sqrt{2b}})$, on general graphs, and
$(3-\frac{2}{\sqrt{b}})$, on bipartite graphs, for the \bmec
problem.
\end{proposition}
\begin{proof}
We call a solution $\langle C \rangle  = \langle C_1, C_2, \dots,
C_k \rangle$ to the \bmec problem \emph{nice} if each color $C_i$,
$1 \leq i \leq k$, is of cardinality $|C_i|=b$ or $C_i$ is maximal
in the subgraph induced by the edges $E \setminus
\bigcup_{j=1}^{i-1} C_j$. We first bound the number of colors in such a solution. \\

\noindent \textit{Claim.} For the number of colors $k$ in any nice
solution to the \bmec problem it holds that:

$\displaystyle \max\{ \Delta, \lceil\frac{|E|}{b}\rceil \} \leq k
\leq \left\{
 \lceil \frac{|E|}{b}\rceil - \lceil\frac{\Delta^2}{2b}\rceil + (2\Delta-1), \mbox{ for general graphs}
\atop
 \lceil\frac{|E|}{b}\rceil - \lceil\frac{\Delta^2}{b}\rceil + (2\Delta-1), \mbox{ for bipartite graphs}
\right.$\\\\

\noindent The lower bounds follow trivially. For the upper bounds,
let $\langle \mathcal{C} \rangle = \langle C_1,C_2,$ $\dots,C_k
\rangle$ be a nice solution, $e=(u,v)$ be an edge in the last
color $C_k$,  and $E_u$ and $E_v$ be the sets of edges adjacent to
vertices $u$ and $v$, respectively. By the niceness of the
solution $\mathcal{C}$ it follows that edge $e$ does not appear in
any color $C_i$, $1 \leq i \leq k-1$, because $|C_i|=b$ or $C_i$
contains at least one edge in $E_u$ or $E_v$. Let $W, X, Y
\subseteq \{C_1,C_2,\dots,C_{k-1} \}$ such that $W=\{C_i:~
|C_i|=b \}$, $X=\{C_i:~ |C_i|< b$ and $C_i$ contains  an edge $e
\in E_u \}$ and $Y=\{C_i:~ |C_i|< b$ and $C_i$ contains  an edge
$e \in E_v \}$. Let $E_1$ be the set of edges in the colors in $W$
and $E_2= E \setminus E_1$ be the set of edges in the colors in $X
\cup Y \cup \{C_k\}$. Then,  $k = \frac{|E_1|}{b}+x+y+1$, where
$x=|X|$ and $y=|Y|$.

Assume, w.l.o.g., that $X_1 Y_1 X_2 Y_2 \dots X_l Y_l$ is the
order of colors in the nice solution $\langle \mathcal{C}
\rangle$, where  $X_i \subseteq X$, $Y_i \subseteq Y$, $1 \leq i
\leq l$, and $X_1$ is possibly empty. Let $x_i = |X_i|$ and $y_i =
|Y_i|$, $1 \leq i \leq l$.\\

\noindent For general graphs, consider a color $C \in X_i$ and let
$S_i = \bigcup_{j=i}^l Y_j$ and $s_i =\sum_{j=i}^l y_j$. The edge
$(u,z) \in C \cap E_u$ prevents at most one edge $(v,z) \in S_i
\cap E_v$ from being into $C$. Moreover, each other edge $(p,q)
\in C$ prevents at most two edges $(v,p), (v,q) \in S_i \cap E_v$
from being into $C$. As the colors in $S_i$ contain exactly $s_i$
edges from $E_v$ and all of them are prevented from being into
color $C$, it follows that $|C| \geq \lceil \frac{s_i-1}{2}\rceil
+1$. Therefore, there exist at least
$x_i\cdot\lceil\frac{s_i-1}{2}\rceil+x_i$ edges in $X_i$. In a
similar way, by considering  a color $C \in Y_i$ there exist at
least $y_i\cdot\lceil\frac{t_i-1}{2}\rceil+y_i$ edges in $Y_i$,
where $t_i=\sum_{j=i+1}^l x_j$. Summing up these bounds, and
taking into account that $y_l - 1 \geq 0$ (since $Y_l$ is not
empty), it follows that

\begin{center}
\begin{tabular}{rl}
$|E_2|$ & $\displaystyle \geq \sum_{i=1}^l  ( x_i\cdot\lceil\frac{s_i-1}{2}\rceil+x_i)
           + \sum_{i=1}^l  ( y_i\cdot\lceil\frac{t_i-1}{2}\rceil+y_i)$\\
        & $\displaystyle = \frac{x(y-1)-(y_1+y_2+\dots+y_{l-1})}{2} +x+y+1$\\
        & $\displaystyle \geq \frac{x(y-1)-(y_1+y_2+\dots+y_{l-1}+y_l-1)}{2} +x+y+1$\\
        & $\displaystyle = \frac{(x-1)(y-1)}{2} +x+y+1
           \geq \frac{(x+1)(y+1)}{2} + 1
           \geq \lceil\frac{(x+1)(y+1)}{2}\rceil$.
\end{tabular}
\end{center}

Therefore,
 $k = \frac{|E \setminus E_2|}{b}+x+y+1
 \leq \lceil\frac{|E|}{b}\rceil - \lceil\frac{(x+1)(y+1)\rceil}{2b}\rceil + x+y+1$
 If $\Delta \leq 2b$ then this quantity is maximized when
 $x=y=\Delta-1$ and hence  $k \leq \lceil\frac{|E|}{b}\rceil - \lceil\frac{\Delta^2}{2b}\rceil + (2\Delta-1)$.
 If $\Delta > 2b$ then the above quantity is maximized when
 $x=\Delta-1$ and $y=0$ and hence $k \leq \lceil\frac{|E|}{b}\rceil - \lceil\frac{\Delta}{2b}\rceil + \Delta
\leq \lceil\frac{|E|}{b}\rceil - \lceil\frac{\Delta^2}{2b}\rceil +
(2\Delta-1)$.\\

\noindent For bipartite graphs, the proof is similar. The
structure of  a bipartite graph allows a tighter bound on the
number of edges in the colors in $X_i$ and $Y_i$. Consider, again,
a color $C \in X_i$. For the edge $(u,z) \in C \cap E_u$, there is
no edge $(v,z) \in S_i \cap E_v$, while each other edge $(p,q) \in
C$ prevents at most one edge $(v,p)$ or $(v,q)$ in $S_i \cap E_v$
from being into $C$. Thus, $|C| \geq s_i+1$ and, hence, there
exist at least $x_i(s_i+1)$ edges in $X_i$. Similarly there exist
at least $y_i(t_i+1)$ edges in $Y_i$. The rest of the proof is
along the same lines, but using these bounds.\\

\noindent We return, now, to the solution, $\langle \mathcal{C}
\rangle = \langle C_1,C_2,\dots,C_k \rangle$, derived by
\algo{Greedy}. Consider the color $C_i$ and let $e_j$ be the first
edge inserted in $C_i$, i.e. $w_i=w(e_j)$. Let
$E_i=\{e_1,e_2,\dots,e_j\}$, $G_i$ be the subgraph of $G$ induced
by the edges in $E_i$, and $\Delta_i$ be the maximum degree of
$G_i$.

The solution $\langle \mathcal{C} \rangle$ is a nice one, since it
is constructed in a First-Fit manner. Moreover, an optimal
solution can be also easily transformed into a nice one of the
same total weight. For general graphs, by the bounds above, it
follows that \emph{(i)} $i \leq \lceil\frac{|E_i|}{b}\rceil -
\lceil\frac{\Delta_i^2}{2b}\rceil + (2\Delta_i-1)$, and
\emph{(ii)} in an optimal solution  the edges of $G_i$ appear in
at least $i^* \geq \max\{\Delta_i,\lceil\frac{|E_i|}{b}\rceil\}$
colors, each one of weight at least $w_i$. Therefore,
$\frac{i}{i^*} \leq \frac{\lceil\frac{|E_i|}{b}\rceil -
\lceil\frac{\Delta_i^2}{2b}\rceil +
(2\Delta_i-1)}{\max\{\Delta_i,\lceil\frac{|E_i|}{b}\rceil\}  }$.
By distinguish between $\Delta_i \geq \lceil\frac{|E_i|}{b}\rceil$
and $\Delta_i < \lceil\frac{|E_i|}{b}\rceil$  it follows that in
either case $\frac{i}{i^*} \leq 3 -
\frac{\Delta_i^2+2b}{2b\Delta_i}$. This bound is maximized when
$\Delta_i = \sqrt{2b}$, that is $\frac{i}{i^*} \leq 3 -
\frac{2}{\sqrt{2b}}$. Thus, $w_i \leq w_{i^*}^* \leq w_{\lceil i /
(3-\frac{2}{\sqrt{2b}}) \rceil}^*$. Summing up these inequalities
for all $i$'s, ~$1 \leq i \leq k$, we obtain the
$(3-\frac{2}{\sqrt{2b}})$ ratio  for general graphs.

A similar analysis yields the $(3-\frac{2}{\sqrt{b}})$ ratio for
bipartite graphs. We present here an example for which the
algorithm performs a ratio of exactly $3-\frac{2}{\sqrt{b}}$ for
bipartite graphs. There is, also, an analogous example for general
graphs. Consider the bipartite graph shown in Figure
\ref{fig:tightness_bipartite}(a), where $C
>> \epsilon$,  and $b=9$. The weight of the optimal solution shown
in Figure \ref{fig:tightness_bipartite}(b) is $3C+3\epsilon$. The
weight of the solution obtained by \algo{Greedy}, shown in Figure
\ref{fig:tightness_bipartite}(c), is $7C -\epsilon$. Thus, the
ratio for this instance is $\frac{7C-\epsilon}{3C+3\epsilon}
\simeq  \frac{7}{3} = 3 - \frac{2}{\sqrt{9}}$.

\begin{figure}[htb]
\begin{center}
\begin{tabular}{ccc}
\multicolumn{3}{c}{\includegraphics{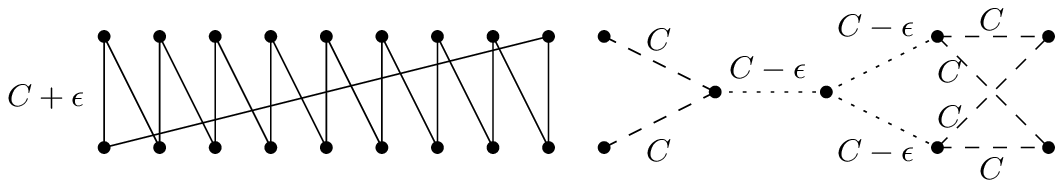}}\\
\multicolumn{3}{c}{(a)}\\
\multicolumn{3}{c}{~}\\
\includegraphics{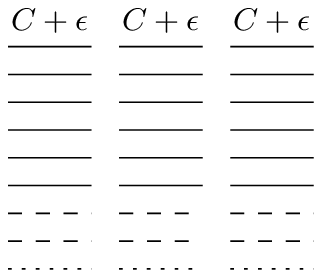} & $~~~~~~~~~~$ & \includegraphics{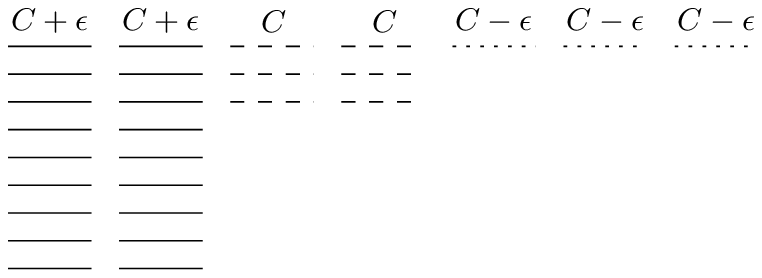}\\
(b)                                                  &              & (c)
\end{tabular}
\end{center}
\caption{(a) An instance of the \bmec problem on bipartite graphs, $C >> \epsilon$, $b=9$.
         (b) An optimal solution.
         (c) The solution obtained by \algo{Greedy}.}
\label{fig:tightness_bipartite}
\end{figure}
\end{proof}

Another approximation result for the \bmec problem is obtained by
exploiting a general framework, presented in \cite{Epstein07},
which allows to convert a $\rho$-approximation algorithm for a
coloring problem into an $e \cdot \rho$-approximation one for the
corresponding max-coloring problem, for hereditary classes of
graphs. In fact, this framework has been presented for such a
conversion from the \vc to the \mvc problem, but it can be easily
seen that this applies also for conversions from the \ec, \bvc and
\bec problems to the \mec, \bmvc and \bmec problems, respectively.
However, this conversion leads to ratios greater than those shown
in Table \ref{table:results} for the \mec and \bmvc problems. For
the \bmec problem on general graphs this approach gives a ratio of
at least $\frac{4}{3} \cdot e > 3$, as the \ec, and hence the
\bec, problem cannot be approximated within a ratio less than
$\frac{4}{3}$. On the other hand, the \bmec problem on bipartite
graphs can be approximated, this way, with a ratio of $e$, as the
\bec problem is polynomial in this case (see Table
\ref{table:results}).

Combining  the discussion above with  Propositions
\ref{prop:bmec_greedy} and \ref{prop:setcover}, it follows that

\begin{theorem}
The \bmec problem can be approximated with a ratio of $\min
\{3-2/\sqrt{2b},H_b\}$, for general graphs, and
$\min\{e,3-2/\sqrt{b},H_b\}$, for bipartite graphs.
\end{theorem}

Note that, the $H_b$ ratio outperforms the other only for $b \leq
5$, for general graphs, and $b=3$, for bipartite graphs, and,
hence, $b$ can be considered as fixed. These ratios are shown in
Table 2, for several values of $b$.

\begin{table}
\label{table:approx}
\begin{center}
\begin{tabular}{| c || l l | l l |}
\hline
~ $b$     ~ & \multicolumn{2}{c |}{General graphs} & \multicolumn{2}{c |}{Bipartite graphs} \\
\hline \hline
~ $3$     ~ & ~ $1.833$ ~~~ & $H_b$           ~    & ~ $1.833$ ~~~ & $H_b$          ~       \\
~ $4$     ~ & ~ $2.083$ ~~~ & $H_b$           ~    & ~ $2.000$ ~~~ & $3-2/\sqrt{b}$ ~       \\
~ $5$     ~ & ~ $2.283$ ~~~ & $H_b$           ~    & ~ $2.106$ ~~~ & $3-2/\sqrt{b}$ ~       \\
~ $6$     ~ & ~ $2.423$ ~~~ & $3-2/\sqrt{2b}$ ~    & ~ $2.184$ ~~~ & $3-2/\sqrt{b}$ ~       \\
~ $\dots$ ~ & ~ $\dots$ ~~~ & $3-2/\sqrt{2b}$ ~    & ~ $\dots$ ~~~ & $3-2/\sqrt{b}$ ~       \\
~ $50$    ~ & ~ $2.800$ ~~~ & $3-2/\sqrt{2b}$ ~    & ~ $2.717$ ~~~ & $3-2/\sqrt{b}$ ~       \\
~ $51$    ~ & ~ $2.802$ ~~~ & $3-2/\sqrt{2b}$ ~    & ~ $2.718$ ~~~ & $e$            ~       \\
~ $\dots$ ~ & ~ $\dots$ ~~~ & $3-2/\sqrt{2b}$ ~    & ~ $\dots$ ~~~ & $e$            ~       \\
\hline
\end{tabular}
\end{center}
\caption{Approximation ratios for the \bmec problem.}
\end{table}


\subsection{NP-completeness for trees}

We prove first that the bounded list edge-coloring, \blec, problem
is NP-complete even if the graph $G=(V,E)$ is a set of chains,
$|\phi(e)| = 2$, for all $e \in E$, and $b=5$. We denote this
problem as EC(chains, $|\phi(e)|=2, ~b=5$).

\begin{proposition}
The EC(chains, $|\phi(e)| = 2, ~b=5$) problem is NP-complete.
\end{proposition}
\begin{proof}
By Theorem \ref{th:list}(i) (\cite{Dror99}), the VC(chains,
$|\phi(v)| \leq 2, ~b_i \leq 5$) problem is NP-complete. Given
that the line-graph of a chain is also a chain, it follows that
the EC(chains, $|\phi(e)| \leq 2, ~b_i \leq 5$) problem is also
NP-complete. The later problem can be easily reduced to the
EC(chains, $|\phi(e)| \leq 2, ~b=5$) problem, where $b_i=b=5$ for
all colors: for every color $C_i$ with $b_i<5$, add $5-b_i$
independent edges with just $C_i$ in their lists. This last
problem reduces to the EC(chains, $|\phi(e)|= 2, ~b=5$) problem,
where $|\phi(e)|=2$ for all edges. This can be done by
transforming an instance of EC(chains, $|\phi(v)| \leq 2, b=5$) as
following: {\em (i)} add two new colors  $C_{k+1}$ and $C_{k+2}$,
both with cardinality bound $b=5$, {\em (ii)} add color $C_{k+1}$
to the list of every edge $e$ with $|\phi(e)|=1$, {\em (iii)} add
ten independent edges and put in their lists both colors $C_{k+1}$
and $C_{k+2}$.
\end{proof}

\begin{theorem}
The \bmec problem on trees is NP-complete.
\end{theorem}
\begin{proof}
Our reduction is from EC(chains, $|\phi(e)|= 2, ~b=5$) problem. We
construct an instance of the \bmec problem on a forest
$G'=(V',E')$ as follows.

We replace every edge $e=(u,v) \in E$ with a chain of three edges:
$e_1=(u,u')$, $e_2=(u',v')$ and $e_3=(v',v)$, where
$w(e_1)=w(e_2)=w(e_3)=1$. Moreover, we create $k-|\phi(e)|=k-2$
stars of $k-1$ edges each. We add edges $(u',s_t)$, $1 \leq t \leq
k-2$, between $u'$ and the central vertex $s_t$ of each of these
$k-2$ stars; thus every star has now exactly $k$ edges. Let
$\phi(e)=\{C_i,C_j\}$. The $k-2$ edges $(u',s_t)$ take different
weights in $\{1,2,\dots,k\} \setminus \{i,j\}$. Let $q$ be the
weight taken by an edge $(u',s_t)$. The remaining $k-1$ edges of
the star $t$ take different weights in $\{1,2,\dots,k\} \setminus
\{q\}$. In the same way, we add $k-2$ stars connected to $v'$. In
Figure \ref{fig:bmec1}, is shown the $u'$'s part of this
edge-gadget for $e=(u,v)$. For every edge $e$ of $G$, we add
$2(k-2)$ stars and $2(k-2)k+2$ edges.
\begin{figure}[htb]
\begin{center}
\includegraphics{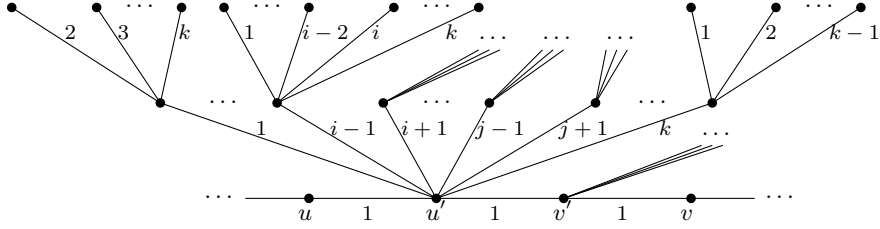}
\end{center}
\caption{The gadget for an edge $e=(u,v)$ with
$\phi(e)=\{C_i,C_j\}$.} \label{fig:bmec1}
\end{figure}

To complete our construction we define $f_i$ to be the frequency
of color $C_i$ in the lists of all edges and $F=\max\{f_i | 1 \leq
i \leq k\}$. For every color $C_i$ we add $F-f_i$ disconnected
copies of the color-gadget shown in Figure \ref{fig:bmec2}.
Such a gadget consists of an edge $e=(x,y)$ and $k-1$ stars with $k-1$ edges each. There are
also edges between one of the endpoints of $e$, say $y$, and the
central vertices of all stars; thus every star has now exactly $k$
edges. The edge $e$ takes weight $i$ and the edges in the stars of such a
color-gadget take weights similarly with those in the stars of an edge-gadget.
For a color $C_i$ we add $(F-f_i)(k-1)$ stars and
$(F-f_i)(k-1)k+(F-f_i)$ edges.
\begin{figure}[htb]
\begin{center}
\includegraphics{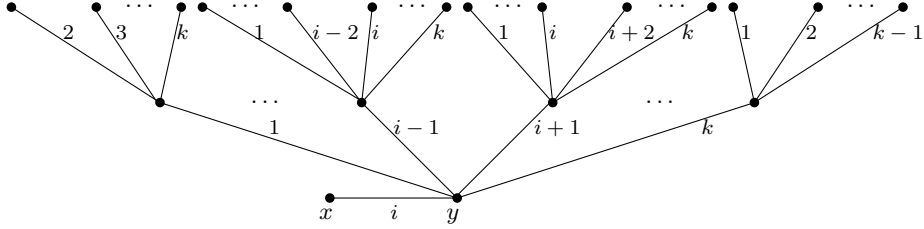}
\end{center}
\caption{A gadget for the color $C_i$.} \label{fig:bmec2}
\end{figure}

The number of stars in the forest $G'$ we have constructed is
$2|E|(k-2)+\sum_{i=1}^k(F-f_i)(k-1) =k(k-1)F-2|E|$, since
$\sum_{i=1}^k f_i=2|E|$. By setting $b'=k(k-1)F-2|E|+5+F$, we
prove that: ``\emph{There is a $k$-coloring for \blec(chains,
$|\phi(e)| = 2, ~b=5$), if and only if, $G'$ has a bounded
max-edge-coloring of total weight $\sum_{i=1}^k i$ such that every
color is used at most $b'$ times}''.

Consider, first, a solution $\mathcal{C}$ to the \blec problem. We
construct a solution $\mathcal{C}'$ for the \bmec problem as
following. Let $e=(u,v) \in E$ be an edge with
$\phi(e)=\{C_i,C_j\}$, which, w.l.o.g., appears in the color $C_i$
of  $\mathcal{C}$. Put the edges $e_1$ and $e_3$ of the
edge-gadget for $e$ in color $C_i'$, while the edge $e_2$ in color
$C_j'$. After doing this for all edges in $E$, each color $C_i'$
contains at most $2 \cdot 5+1 \cdot (f_i-5) = f_i+5$ edges. Next,
put the edges with weight $i$, $1 \leq i \leq k$, from the
$k(k-1)F-2|E|$ stars into $C_i'$. Each color $C_i'$ in
$\mathcal{C}'$ constructed so far contains at most
$k(k-1)F-2|E|+f_i+5=b'-(F-f_i)$ edges and, by the construction of
$G'$, $\mathcal{C}'$ is a proper coloring. In the $F-f_i$
color-gadgets for $C_i$ there are $F-f_i$ remaining $(x,y)$ edges
of weight $i$, which can still be inserted into color $C_i'$.
Thus, we get a solution for the \bmec problem of $k$ colors, each
one of at most $b'$ edges, and total weight $\sum_{i=1}^k i$.

Conversely, consider a solution $\mathcal{C}'$ to the \bmec
problem. $\mathcal{C}'$ consists of exactly $k$ colors of weights
$1,2,\dots,k$, since each star in $G'$ has $k$ edges and each edge
has a different weight in the range $\{1,2,\dots,k\}$. Thus, all
edges of the same weight, say $i$, should belong in the same color
$C_i'$ of $\mathcal{C}'$. Therefore, $C_i'$ contains one edge from
each one of the $k(k-1)F-2|E|$ stars as well as the $F-f_i$
remaining $(x,y)$ edges of the color-gadgets having weight $i$.
Consider, now, the edges of $G'$ corresponding to the edges $e_1$,
$e_2$ and $e_3$ of the edge-gadget for an edge $e$ with
$\phi(e)=\{C_i,C_j\}$. By the construction of $G'$ and the choice
of edge weights, the edges $e_1$, $e_2$ and $e_3$ should appear
into colors $C_i'$ and $C_j'$. Thus, edges  $e_1$ and $e_3$ should
appear, w.l.o.g., into color $C_i'$, while $e_2$ into color
$C_j'$. Therefore, the edge $e \in E$ can be colored by color $C_i
\in \phi(e)$. Finally, a color $C_i'$ contains at most $5$ edges
of type $e_1$ (or $e_3$), corresponding to at most $5$ edges of
$E$; otherwise $|C_i'| \geq k(k-1)F-2|E|+(F-f_i)+(2 \cdot 6 + 1
\cdot (f_i-6))>b'$, a contradiction.

To complete our proof for the \bmec problem on trees, let $p$ be
the number of trees in $G'$. We add a set of $p-1$ edges of weight
$\epsilon < 1$ to transform the forest $G'$ into a single tree
$T$. This can be easily done since every tree of $G'$ has at least
two vertices. By keeping the same bound $b'$, it is easy to see
that there is a solution for the \bmec problem on $G'$ of weight
$\sum_{i=1}^k i$, if and only if, there is a solution for the
\bmec problem on $T$ whose weight is equal to $\sum_{i=1}^k
i+\lceil\frac{p-1}{b'}\rceil \epsilon$.
\end{proof}

\subsection{A 2-approximation algorithm for trees}

In \cite{Lucarelli09} a 2-approximation algorithm for the \mec
problem on trees has been presented,
 which is also exploited to derive a ratio of 3/2 for that problem.  This algorithm yields
to a solution of $\Delta$ colors, $\mathcal{M} =
\{M_1,M_2,\dots,M_{\Delta} \}$. Starting from this solution we
obtain a solution to the \bmec problem by finding the ordered
$b$-partition of each color in $\mathcal{M}$. For the sake of
completeness we give below the whole algorithm.\\

\noindent \algo{Convert} \texttt{
\\1. Let $T_r$ be the tree rooted in an arbitrary vertex $r$;
\\2. For each vertex $v$ in pre-order traversal of $T_r$ do
\\3. $~$ Let $\langle E_v \rangle = \langle e_1, e_2, \dots, e_{d(v)}\rangle$ be the edges adjacent to $v$,
\\ $~~~~~$ and $(v,p)$ be the edge from $v$, $v \not \equiv r$, to its parent;
\\4. $~$ Using ordering $\langle E_v \rangle$, insert each edge in $E_v$, but $(v,p)$,
\\ $~~~~~$ into the first matching which does not contain an edge in $E_v$;
\\5. Let $\mathcal{M} =\{M_1,M_2,\dots,M_{\Delta}\}$ be the colors constructed;
\\6. For $i=1$ to $\Delta$ do
\\7. $~$ Let $\mathcal{P}_{M_i}=\{M^i_1,M^i_2,\dots,M^i_{k_i}\}$ be the ordered $b$-partition of $\langle M_i \rangle$;
\\8. Return a solution $\langle\mathcal{C}\rangle = \langle C_1,C_2,\dots,C_k \rangle$, $\mathcal{C} = \bigcup_{i=1}^{\Delta} \mathcal{P}_{M_i}$;
}

\begin{theorem}
\algo{Convert} is a $2$-approximation one for the \bmec problem on
trees.
\end{theorem}
\begin{proof}
Consider the color $C_j$  in the solution $\langle \mathcal{C}
\rangle$ and let $e$ be the heaviest edge in  $C_j$, i.e.,
$w(e)=w_j$. Let $X \subseteq \mathcal{C}_j
=\{C_1,C_2,\dots,C_{j-1}\}$ such that each color $C_p \in X$ has
\emph{(i)} $|C_p|=b$, and \emph{(ii)} all edges of weight at least $w(e)$. Let
also $Y =\mathcal{C}_j\setminus X$, $|X|=x$ and $|Y|=y$. Clearly,
$x+y=j-1$. Let $j^*$ be the number of colors in an optimal
solution of weight at least $w(e)$, that is $w_j=w_{j^*}^*$.

There are at least $x \cdot b+y+1$ edges of weight at least
$w(e)$. These edges in an optimal solution appear in at least
$\lceil\frac{x \cdot b+y+1}{b}\rceil \geq x+1$ colors, that is,
$j^* \geq x+1$.

We show, next, that all colors in $Y \cup \{C_j\}$ come from $y+1$
different colors in $ \mathcal{M}$. Assume that two of these
colors, $C_q$ and $C_r$, come from the ordered $b$-partition of
the same color $M_t \in \mathcal{M}$. Assume, w.l.o.g., that $w_q
\geq w_r$, and let $f$ be the heaviest edge in $C_r$. Note that
$C_r$ may coincide with $C_j$, while $C_q$ cannot. As $C_q \in Y$,
it follows that $|C_q|=b$ and there is an edge $f' \in C_q$ with
$w(f') < w(e) \leq w(f)$, a contradiction to the definition of the
ordered $b$-partition of $M_t$. Therefore, $C_j$ comes from a
color $M_i \in \mathcal{M}$, $i \geq y+1$, that is $e \in M_i$. By
the construction of the coloring $\mathcal{M}$, there are at least
$i-1$ edges, adjacent to each other, of weight at least $w(e)$
(i.e., $i-2$ of them adjacent to $e$ and $e$ itself). These $i-1$
edges appear in different colors in an optimal solution, that is,
$j^* \geq y$.

Combining the two lower bounds for $j^*$ and taking into account
that $x+y=j-1$ we get $j^* \geq \lceil\frac{j}{2}\rceil$.
Therefore, $w_j=w_{j^*}^* \leq w_{\lceil\frac{j}{2}\rceil}^*$ and
summing up the weights $w_j$ of all colors in $\mathcal{C}$ we get
$W = \sum_{j=1}^k w_j \leq 2\sum_{j=1}^{\lceil k/2 \rceil} w_j^*
\leq 2\sum_{j=1}^{k^*} w_j^* \leq 2OPT$, since $k^* \geq \lceil
\frac{k}{2} \rceil$. A tight example for this algorithm, is given
in  \cite{Lucarelli09}, as for large values of $b$ the \bmec
coincides with the  \mec problem. By a careful analysis, the
complexity of both Lines 2 and 6 of the algorithm is $O(n\log n)$,
where $n$ is the number of vertices of the tree.
\end{proof}


\bibliographystyle{plain}
\bibliography{BMC}

\end{document}